\documentclass[runningheads]{llncs}
\usepackage{graphicx}
\usepackage{amsmath}
\usepackage[utf8]{inputenc}
\usepackage[lined,algo2e,vlined,ruled, linesnumbered]{algorithm2e}
\usepackage[normalem]{ulem}
\usepackage{thm-restate}
\usepackage{tikz}
\usetikzlibrary{arrows, positioning,decorations.pathreplacing, tikzmark,shapes}
\usetikzlibrary{plotmarks}

\DeclareMathAlphabet{\mathcal}{OMS}{cmsy}{m}{n}
\DeclareMathAlphabet{\bm}{OT1}{ptm}{b}{it}

\newcommand{\probl}{\text{MCPC}}
\newcommand{\cnp}{\text{NP}}
\newcommand{\cp}{\text{P}}
\newcommand{\val}{p}

\newcommand{\bgt}{B}
\newcommand{\cst}{c}
\newcommand{\ccap}{U}
\newcommand{\iset}{\mathcal{I}}
\newcommand{\sset}{\mathcal{S}}
\newcommand{\kset}{\mathcal{K}} 
\newcommand{\cset}{\mathcal{C}} 
\newcommand{\alg}{\ensuremath{\textsc{alg}}} 
\newcommand{\greedy}{\ensuremath{\textsc{greedy}}} 
\newcommand{\res}{\ensuremath{\textit{res}}}
\newcommand{\opt}{\ensuremath{\textsc{opt}}} 
\newcommand{\nkincarg}[1]{\ensuremath{q(#1)}} 
\newcommand{\crit}[1]{\ensuremath{\kappa}(#1)} 
\newcommand{\kcu}{\ensuremath{\textit{US}}}
\newcommand{\kcs}{\ensuremath{\textit{SS}}}
\newcommand{\usarg}[1]{\ensuremath{\kcu(#1)}}
\newcommand{\ssarg}[1]{\ensuremath{\kcs(#1)}}

\newcommand{\iso}{\ensuremath{\iota}}

\newcommand{\splits}{\ensuremath{\varsigma}}
\newcommand{\ip}{\ensuremath{\textrm{IP}}}
\newcommand{\lp}{\ensuremath{\textrm{LP}}}
\newcommand{\seq}[1]{\ensuremath{\langle#1\rangle}}
\newcommand{\vneg}{\par\vspace*{-1\bigskipamount}}
\newcommand{\ssetarg}[1]{\ensuremath{\sset(#1)}}
\newcommand{\ksetarg}[1]{\ensuremath{\kset(#1)}}
\newcommand{\csetarg}[1]{\ensuremath{\cset(#1)}}
\newcommand{\cov}[1]{\ensuremath{\cup#1}}
\newcommand{\mset}[1]{\ensuremath{\{#1\}}}
\newcommand{\msset}[2]{\ensuremath{\{ #1 \, | \, #2\}}}

\begin{document}
\title{Maximum Coverage with Cluster Constraints: An LP-Based Approximation Technique}

\titlerunning{Maximum Coverage with Cluster Constraints}
\author{Guido Sch\"afer\inst{1,2} \and\\
Bernard G. Zweers\inst{1} }

\authorrunning{G.Sch\"afer and B.G. Zweers.}
\institute{Centrum Wiskunde \& Informatica (CWI), Amsterdam, the Netherlands\\
\email{\{g.schaefer,b.g.zweers\}@cwi.nl}
\and
Vrije Universiteit, Amsterdam, the Netherlands}

\maketitle              

\begin{abstract}
Packing problems constitute an important class of optimization problems, both because of their high practical relevance and theoretical appeal. However, despite the large number of variants that have been studied in the literature, most packing problems encompass a single tier of capacity restrictions only. For example, in the \emph{Multiple Knapsack Problem}, we want to assign a selection of items to multiple knapsacks such that their capacities are not exceeded. But what if these knapsacks are partitioned into \emph{clusters}, each imposing an additional (aggregated) capacity restriction on the knapsacks contained in that cluster? 

In this paper, we study the \emph{Maximum Coverage Problem with Cluster Constraints (\probl)}, which generalizes the \emph{Maximum Coverage Problem with Knapsack Constraints (MCPK)} by incorporating such cluster constraints. Our main contribution is a general LP-based technique to derive approximation algorithms for such cluster capacitated problems. Our technique basically allows us to reduce the cluster capacitated problem to the respective original packing problem (i.e., with knapsack constraints only). By using an LP-based approximation algorithm for the original problem, we can then obtain an effective rounding scheme for the problem, which only loses a small fraction in the approximation guarantee. 

We apply our technique to derive approximation algorithms for \probl. To this aim, we develop an LP-based $\frac{1}{2}(1-\frac1e)$-approximation algorithm for MCPK by adapting the \emph{pipage rounding technique}. Combined with our reduction technique, we obtain a $\frac{1}{3}(1-\frac{1}{e})$-approximation algorithm for \probl. 
We also derive improved results for a special case of \probl, the \emph{Multiple Knapsack Problem with Cluster Constraints (MKPC)}. Based on a simple greedy algorithm, our approach yields a $\frac13$-approximation algorithm. By combining our technique with a more sophisticated iterative rounding approach, we obtain a $\frac{1}{2}$-approximation algorithm for certain special cases of MKPC. 

\end{abstract}

\section{Introduction}

Many optimization problems encountered in real-life can be modeled as a \emph{packing problem}, where a set of diverse objects (or items, elements) need to be packed into a limited number of containers (or knapsacks, bins) such that certain feasibility constraints are satisfied. Given their high practical relevance, there is a large variety of packing problems that have been studied in the literature. For example, the monograph by Kellerer et al.~\cite{Kellerer2004} discusses more than 20 variants of the classical \emph{Knapsack Problem (KP)} alone. 

However, despite the large number of variants that exist, most packing problems encompass a \emph{single} tier of capacity restrictions only. For the sake of concreteness, consider the \emph{Multiple Knapsack Problem (MKP)}: In this problem, the goal is to find a most profitable selection of items that can be assigned to the (multiple) knapsacks without violating their capacities, i.e., there is a single capacity constraint per knapsack that needs to be satisfied. But what about the setting where these knapsacks are partitioned into \emph{clusters}, each imposing an additional (aggregated) capacity restriction on all knapsacks contained in it? 

We believe that the study of such problems is well motivated and important, both because of their practical relevance and theoretical appeal. For example, a potential application of the \emph{Multiple Knapsack Problem with Cluster Constraints (MKPC)} (as outlined above) is when items have to be packed into containers, obeying some weight capacities, and these boxes have to be loaded onto some weight-capacitated vehicles (e.g., ships). Another example is a situation in which customers have to be assigned to facilities, which have a certain capacity, but these facilities are also served by larger warehouses that again have a restricted capacity.

\subsection{Maximum Coverage with Cluster Constraints}
In this paper, we propose and initiate the study of such extensions for packing problems. More specifically, we consider the following packing problem, which we term the \emph{Maximum Coverage Problem with Cluster Constraints (\probl)} (see Section~\ref{s:prelim} for formal definition): Basically, in this problem we are given a collection of subsets of items, where each subset is associated with some cost and each item has a profit, and a set of knapsacks with individual capacities. In addition, the knapsacks are partitioned into clusters which impose additional capacity restrictions on the total cost of the subsets assigned to the knapsacks in each cluster. The goal is to determine a feasible assignment of a selection of the subsets to the knapsacks such that both the knapsack and the cluster capacities are not exceeded, and the total profit of all items covered by the selected subsets is maximized. 

As we detail in the related work section below, \probl\ is related to several other packing problems but has not been studied in the literature before (to the best of our knowledge). It generalizes several fundamental packing problems such as the \emph{Maximum Coverage Problem with Knapsack Constraints (MCPK)} (see, e.g., \cite{Ageev2004,Khuller99}), which in turn is a generalization of the Multiple Knapsack Problem (MKP). 
Another important special case of \probl\ is what we term the \emph{Multiple Knapsack Problem with Cluster Constraints (MKPC)} (as outlined above). Also this problem has not been addressed in the literature before and will be considered in this paper. 

All problems considered in this paper are (strongly) \cnp-hard (as they contain MKP as a special case), which rules out the existence of a polynomial-time algorithm to compute an optimal solution to the problem (unless $\cnp = \cp$).
We therefore resort to the development of approximation algorithms for these problems. Recall that an \emph{$\alpha$-approximation algorithm} computes a feasible solution recovering at least an $\alpha$-fraction ($0 \le \alpha \le 1$) of the optimal solution in polynomial time. Note that the \emph{Maximum Coverage Problem} is a special case of \probl\ and this problem is known to be $(1-\frac{1}{e})$-inapproximable (see \cite{Feige1998}). As a consequence, we cannot expect to achieve an approximation factor better than this for \probl.

We remark that while our focus here is mostly on the extensions of MCPK and MKP, the idea of imposing a second tier of capacity restrictions is generic and can be applied to other problems as well (not necessarily packing problems only).\footnote{In fact, some preliminary results (omitted from this paper) show that our technique can also be applied to derive an approximation algorithm for the capacitated facility location problem with cluster constraints.}

\subsection{Our Contributions}%
We present a general technique to derive approximation algorithms for packing problems with two tiers of capacity restrictions. Basically, our idea is to extend the natural integer linear programming (ILP) formulation of the original packing problem (i.e., without cluster capacities) by incorporating the respective cluster capacity constraints. But, crucially, these new constraints are set up in such a way that an optimal solution to the LP relaxation of this formulation defines some \emph{reduced capacities} for each knapsack individually. This enables us to reduce the cluster capacitated problem to the respective original packing problem with knapsack constraints only. We then use an LP-based approximation algorithm for the original problem to round the optimal LP-solution. This rounding requires some care because of the reduced capacities (see below for details). 

Here we apply our technique to derive approximation algorithms for \probl\ and MKPC. As mentioned, to this aim we need LP-based approximation algorithms for the respective problems without cluster capacities, namely MCPK and MKP, respectively. While for the latter a simple greedy algorithm gives a $\frac{1}{2}$-approximation, we need more sophisticated techniques to derive an LP-based approximation algorithm for MCPK. In particular, we adapt the \emph{pipage rounding technique} \cite{Ageev2004} to obtain an LP-based $\frac{1}{2}(1-\frac{1}{e})$-approximation for this problem (Section \ref{ss:MCKP}), which might be of independent interest; this is also one of our main technical contributions in this paper.
Based on these algorithms, we then obtain a $\frac{1}{3}(1-\frac{1}{e})$-approximation algorithm for \probl\ (Section \ref{ss:MCPC}) and a $\frac{1}{3}$-approximation algorithm for MKPC (Section \ref{s:additionalApplications}).
Finally, we show that by combining our technique with a more sophisticated iterative rounding approach, we can obtain an improved $\frac{1}{2}$-approximation algorithm for a special case of the MKPC (Section~\ref{s:additionalApplications}); this part is another main technical contribution of the paper. Our iterative rounding approach applies whenever the clusters satisfy a certain \emph{isolation property} (which is guaranteed, for example, if the clusters can be disentangled in a natural way). 

\subsection{Related Literature}%
The coverage objective that we consider in this paper is a special case of a submodular function and there is a vast literature concerning the problem of maximizing a (monotone) submodular function. Nemhauser and Wolsey~\cite{Nemhauser1978} were the first to study this problem under a cardinality constraint. 
The authors propose a natural greedy algorithm to solve this problem and showed that it achieves a $(1-\frac{1}{e})$-approximation guarantee. Later, Feige~\cite{Feige1998} shows that the factor of $(1-\frac{1}{e})$ is best possible for this problem (unless $\cp= \cnp$). 

Khuller et al.~\cite{Khuller99} use a modification of the greedy algorithm of \cite{Nemhauser1978} in combination with partial enumeration to obtain a $(1-\frac{1}{e})$-approximation algorithm for the Maximum Coverage Problem with a budget constraint. 
Ageev and Sviridenko \cite{Ageev2004} introduce the technique of \emph{pipage rounding} and also derive a $(1-\frac{1}{e})$-approximation algorithm for this problem. Sviridenko \cite{Sviridenko2004} observed that the algorithm of \cite{Khuller99} can be applied to submodular functions and achieves a $(1-\frac1e)$-approximation ratio.
Badanidiyuru and Vondr{\'a}k \cite{Badanidiyuru2014} use another approach for maximizing a submodular function given a knapsack constraint. They derive a $(1-\frac{1}{e}-\varepsilon)$-approximation algorithm whose running time decreases in $\varepsilon>0$. 
Ene and Nguyen \cite{Ene2019} show that there are some technical issues with the approach of Badanidiyuru and Vondr{\'a}k \cite{Badanidiyuru2014}, but they build upon their idea to derive a faster algorithm with the same approximation ratio.  
For non-monotone submodular, a $(1-\frac{1}{e} -\varepsilon)$-approximation algorithm is given by Kulik et al.~ \cite{Kulik2013}.

Further, the problem of maximizing a submodular function subject to multiple knapsack constraints has also been studied (see \cite{Chekuri2014,Kulik2013,Lee2010}): 
Kulik et al.~\cite{Kulik2013} obtain a  $(1-\frac{1}{e} -\varepsilon)$-approximation algorithm (for any $\varepsilon>0$) for the Maximum Coverage Problem with a $d$-dimensional knapsack constraint. The technique also extends to non-monotone submodular maximization with a $d$-dimensional knapsack constraint.
Lee et al.~\cite{Lee2010} give a $(\frac{1}{5}-\varepsilon)$-approximation algorithm for non-monotone submodular maximization with a $d$-dimensional knapsack. The randomized rounding technique that is used in \cite{Lee2010} is similar to that of \cite{Kulik2013}, but the algorithm to solve the fractional relaxation is different. 
However, in all three works \cite{Chekuri2014,Kulik2013,Lee2010}, the interpretation of multiple knapsack constraints is different from the one we use here. In particular, in their setting the costs of the sets and the capacity of a single knapsack are $d$-dimensional and a feasible assignment needs to satisfy the capacity in every dimension. On the other hand, in our definition there are multiple knapsacks, but their capacity is only one-dimensional. Hence, the techniques in \cite{Chekuri2014,Kulik2013,Lee2010} do not apply to our problem.

Simultaneously to our work, two papers have appeared that also discuss the MCPK (see \cite{Fairstein2020}  and \cite{Sun2020}). Fairstein et al.~ \cite{Fairstein2020}, present a randomized $\left(1-\frac{1}{e}-\varepsilon\right)$-approximation algorithm for the MCPK in which a greedy approach for submodular maximization is combined with a partitioning of knapsacks in groups of approximately the same size. In Sun et al.~\cite{Sun2020}, a deterministic greedy $\left(1-\frac{1}{e}-\varepsilon\right)$-approximation algorithm is given for the case in which all knapsacks have the same size. For the general case, their deterministic algorithm has an approximation ratio of $\frac{1}{2} - \varepsilon$, and they present a randomized  $\left(1-\frac{1}{e}-\varepsilon\right)$-approximation algorithm.

Another problem that is related to our \probl\ is the Cost-Restricted Maximum Coverage Problem with Group Budget Constraints (CMCG) introduced by Chekuri and Kumar~\cite{Chekuri2004}. In this problem, there are predefined groups that are a union of sets. Each group has a budget that must not be exceeded by its assigned sets, and there is a single knapsack constraint for all selected sets. The authors give a greedy algorithm that achieves an approximation ratio of $\frac{1}{12}$. Besides this cost-restricted version, Chekuri and Kumar~\cite{Chekuri2004} also study the cardinality-restricted version for the number of sets assigned to a group. For this problem, they obtain a $\frac{1}{\alpha+1}$-approximation algorithm that uses an oracle that, given a current solution, returns a set which contribution to the current solution is within a factor $\alpha$ of the optimal contribution of a single set. For the cost-restricted version, Farbstein and Levin \cite{Farbstein2017}  obtain a $\frac{1}{5}$-approximation. Guo et al.~\cite{Guo2019} present a pseudo-polynomial algorithm for CMCG whose approximation ratio is $(1-\frac{1}{e})$. 
The \probl\ with a single cluster is a special case CMCG, in which each group correspond to a knapsack and each set has a copy for each feasible knapsack. If the solution for a single cluster is seen as a new set, then the \probl\ with multiple clusters can be seen as the cardinality restricted version in which each cluster correspond to a group to which at most one set can be assigned. Combining the approaches of Farbstein and Levin~\cite{Farbstein2017} and Chekuri and Kumar~\cite{Chekuri2004} for the cost-restricted and cardinalty-restricted version, respectively, we obtain a $\frac{1}{6}$-approximation algorithm. In this paper, we improve this ratio to $\frac{1}{3}\left(1-\frac{1}{e}\right)$.

Finally, we elaborate on the relationship between MKPC and closely related problems. We focus on three knapsack problems that look similar to MKPC but are still slightly different. 
Nip and Wang~\cite{Nip2018} consider the \emph{Two-Phase Knapsack Problem (2-PKP)} and obtain a $\frac{1}{4}$-approximation algorithm for this problem. in 2-PKP, items are assigned to knapsacks and the full knapsack needs to be packed in a cluster. Note that MKPC and 2-PKP are different because in MPKC only the costs assigned to a knapsack are restricted by the cluster capacity, whereas in 2-PKP each knapsack contributes with its maximum budget to the cluster. 
Dudzinski and Walukiewucz~\cite{Dudzinski1987} study the \emph{Nested Knapsack Problem (NKP)}, where there are multiple subsets of items and each subset has a capacity. That is, in NKP there are no predefined knapsacks to which the items can be assigned, but from each set we have to select the most profitable items. If these sets are disjoint, the problem is called the \emph{Decomposed Knapsack Problem}. 
The authors present an exact branch-and-bound method, based on the Lagrangean relaxation.
Xavier and Miyazawa~\cite{Xavier2006} consider the \emph{Class Constraint Shelf Knapsack Problem (CCSKP)}. Here, there is one knapsack with a certain capacity and the items assigned to the knapsack should also be assigned to a shelf. The shelf has a maximum capacity, and between every two shelves a divisor needs to be placed. In the CCSKP, each item belongs to a class and each shelf must only have items of the same class. The authors derive a PTAS for CCSKP in \cite{Xavier2006}.

\section{Preliminaries}
\label{s:prelim}

The \emph{Maximum Coverage Problem with Cluster Constraints (\probl)} considered in this paper is defined as follows: We are given a collection of $m$ subsets $\sset = \mset{S_1, \dots, S_m}$ over a ground set of items $\iset=[n]$.%
\footnote{Throughout the paper, given an integer $n \ge 1$ we use $[n]$ to refer to the set $\mset{1, \dots, n}$.}
Each subset $S_j \subseteq \iset$, $j \in [m]$, is associated with a cost $\cst_j > 0$ and each item $i \in \iset$ has a profit $\val_i > 0$. 
For notational simplicity, we identify $\sset = [m]$, and for every subset $\sset' \subseteq \sset$, define $\cov{\sset'} = \cup_{j \in \sset'} S_j \subseteq \iset$ as the set of items covered by $\sset'$. 
Further, we use $\ssetarg{i} \subseteq \sset$ to refer to the subsets in $\sset$ that contain item $i \in \iset$, i.e., $\ssetarg{i} = \msset{j \in \sset}{i \in S_j}$.
In addition, we are given a set of knapsacks $\kset = [p]$ and each knapsack $k \in \kset$ has a capacity $\bgt_k > 0$. These knapsacks are partitioned into a set of $q$ clusters $\cset = [q]$ and each cluster $l \in \cset$ has a separate capacity $\ccap_l > 0$. We denote by $\ksetarg{l}\subseteq \kset$ the subset of knapsacks that are contained in cluster $l$. The number of knapsacks in the set $\ksetarg{l}$ is given by $\nkincarg{l}:=|\ksetarg{l}|$. Also, we use $\csetarg{k} \in \cset$ to refer to the cluster containing knapsack $k \in \kset$. 

Our goal is to determine a feasible assignment $\sigma: \sset \rightarrow \kset \cup \mset{0}$ of subsets to knapsacks such that the total profit of all covered items is maximized. Each subset $j \in \sset$ can be assigned to at most one knapsack in $\kset$ and we define $\sigma(j) = 0$ if $j$ remains unassigned.  
We say that an assignment $\sigma$ is \emph{feasible} if 
(i) for every knapsack $k \in \kset$, the total cost of all subsets assigned to $k$ is at most $\bgt_k$, and 
(ii) for every cluster $l \in \cset$, the total cost of all subsets assigned to the knapsacks in $\ksetarg{l}$ of cluster $l$ is at most $\ccap_l$.
Given an assignment $\sigma$, let $\sset_\sigma(k) \subseteq \sset$ be the set of subsets assigned to knapsack $k \in \kset$ under $\sigma$, and let $\sset_\sigma = \cup_{k \in \kset} \sset_\sigma(k) \subseteq \sset$ be the set of all assigned subsets.
Using the notation above, the \probl\ problem is formally defined as follows: 
\[
\max_\sigma 
\Bigg\{ 
\sum_{i \in \cov{\sset_\sigma}} p_i \; \bigg| \;
\sum_{j \in \sset_\sigma(k)} c_j \le \bgt_k \;\; \forall k \in \kset,
\sum_{k \in \ksetarg{l}} \sum_{j \in \sset_\sigma(k)} c_j \le \ccap_l \;\; \forall l \in \cset
\Bigg\}.
\]

We make the assumptions stated below throughout the paper. 
Note that these assumptions are without loss of generality: If the first condition is violated by some $k \in \ksetarg{l}$ then we can simply redefine $\bgt_k = \ccap_l$. If the second condition is not satisfied then the capacity of cluster $l$ is redundant because each knapsack in cluster $l$ can then contain its maximum capacity. 
We assume that for every cluster $l \in \cset$: (i) for every knapsack $k \in \ksetarg{l}$, $\bgt_k \le \ccap_l$, and (ii) $\sum_{k\in \ksetarg{l}}\bgt_k > \ccap_l$.

The special case of \probl\ in which all cluster capacities are redundant is called the \emph{Maximum Coverage Problem with Knapsack Constraints (MCPK)}. 
Another special case of MCPK is the classical \emph{Multiple Knapsack Problem (MKP)} which we obtain if $\sset = [n]$ and $S_i = \mset{i}$ for all $i \in \sset$. 
We refer to the generalization of MKP with (non-redundant) cluster capacities as the \emph{Multiple Knapsack Problem with Cluster Constraints (MKPC)}.

\section{Maximum Coverage Problem with Knapsack Constraints}
\label{ss:MCKP}

We start by deriving an approximation algorithm for the Maximum Coverage Problem with Knapsack Constraints (MCPK). 
The following is a natural integer linear programming (ILP) formulation:
\begin{subequations}
\begin{alignat}{3}
\max \quad 
& & L(x) = \textstyle\sum_{i \in \iset} \val_iy_i 						&				& & \label{eq:objILP} \\
\text{subject to} \quad 
& & 	\textstyle\sum_{j\in \sset} \cst_jx_{jk} 					& \leq \bgt_k 		& & \forall k \in \kset 		\label{eq:BudgetILP} \\
& & 	\textstyle\sum_{k \in \kset}x_{jk} 					& \leq 1 	   		& & \forall j \in \sset 		\label{eq:max1assigned} \\
& &	\textstyle\sum_{j\in \ssetarg{i}}\sum_{k \in \kset}x_{jk}	& \geq y_{i} 		& & \forall i\in\iset 		\label{eq:XYcoupleILP} \\
& & 	x_{jk} 										& \in \{0,1\} \qquad 	& & \forall j \in \sset, \ \forall k \in \kset \label{eq:intx} \\
& & 	y_i 											& \in \{0,1\} 		& & \forall i \in \iset \label{eq:intILP}
\end{alignat}
\end{subequations}
The decision variable $x_{jk}$ indicates whether set $j\in \sset$ is assigned to knapsack $k\in \kset$. Constraint \eqref{eq:BudgetILP} ensures that the total cost assigned to each knapsack is at most its budget and constraint \eqref{eq:max1assigned} makes sure that each set is assigned to at most one knapsack. 
In addition, the decision variable $y_i$ indicates whether item $i \in \iset$ is covered. Constraint \eqref{eq:XYcoupleILP} ensures that an item can only be covered if at least one of the sets containing it is assigned to a knapsack. 
We refer to above ILP \eqref{eq:objILP}--\eqref{eq:intILP} as (\ip) and to its corresponding linear programming relaxation as (\lp).
It is important to realize that for any feasible fixing of the $x_{jk}$-variables (satisfying \eqref{eq:BudgetILP} and \eqref{eq:max1assigned}), the optimal $y_i$-variabels are easily determined by setting $y_i =\min \{1, \sum_{j\in \ssetarg{i}}\sum_{k\in \kset}x_{jk}\}$. As a consequence, an optimal solution $(x, y)$ of the above program is fully specified by the corresponding $x$-part.%
\footnote{\label{rem:implicit-restrictions}We comment on a subtle but important point here: Note that a set $j \in \sset$ cannot be assigned to a knapsack $k \in \kset$ whenever $\cst_j > \bgt_k$. While these restrictions are taken care of implicitly in the ILP formulation (\ip), they are not in the LP-relaxation (\lp). In fact, we would have to add these restrictions explicitly by defining variables $x_{jk}$ only for all $j \in \sset$ and $k \in \kset(j)$, where $\kset(j) \subseteq \kset$ is the set of knapsacks whose capacity is at least $\cst_j$, and adapt the constraints accordingly. However, these adaptations are straight-forward, and for notational convenience, we do not state them explicitly. In the remainder, whenever we refer to (\lp) our understanding is that we refer to the corresponding LP-relaxation with these assignment restrictions incorporated.  No confusion shall arise.}

We crucially exploit that we can find a solution to the LP-relaxation which satisfies the \emph{bounded split property}.
Let $x$ be a fractional solution of (\lp). We define the \emph{support graph} $H_x$ of $x$ as the bipartite graph $H_x = (\sset \cup \kset, E_x)$ with node sets $\sset$ and $\kset$ on each side of the bipartition, respectively, and the edge set $E_x$ contains all edges $\mset{j,k}$ for which $x_{jk}$ is non-integral, i.e., 
$E_x = \msset{ \mset{j, k} \in \sset \times \kset}{0 < x_{jk} < 1}$.
Let $M \subseteq E_x$ be a matching of the support graph $H_x$.\footnote{Recall that a matching $M \subseteq E_x$ is a subset of the edges such that no two edges in $M$ have a node in common.} 
$M$ \emph{saturates} all nodes in $\sset$ if for every (non-isolated) node $j \in \sset$ there is a matching edge $\mset{j, k} \in M$; we also say that $M$ is an \emph{$\sset$-saturating matching} of $H_x$. 
A feasible solution $x$ of (\lp) satisfies the \emph{bounded split property} if the support graph $H_x$ of $x$ has an $\sset$-saturating matching.

For MKP, it is known that there exists an optimal solution $x^*$ of the LP-relaxation for which the bounded split property holds (see, e.g., \cite{Chekuri2005,Shmoys1993}). Deriving a $\frac12$-approximation algorithm for the problem is then easy: The idea 
is to decompose $x^*$ into an integral and a fractional part. The integral part naturally corresponds to a feasible integral assignment and, exploiting the bounded split property, the corresponding $\sset$-saturating matching gives rise to another feasible integral assignment. Thus, by taking the better of the two assignments, we recover at least half of the the optimal LP-solution. 
Unfortunately, for our more general MCPK problem the optimal solution of (\lp) does not necessarily satisfy the bounded split property.

Instead, below we show that there always exists a solution to the LP-relaxation of MCPK which satisfies the bounded split property and is only a factor $1-\frac{1}{e}$ away from the optimal solution. 
We use the technique of \emph{pipage rounding} \cite{Ageev2004}.
Define a new program (CP) as follows: 
\begin{subequations}
\begin{alignat}{3}
\max \quad 
& & F(x) =  \textstyle\sum_{i \in \iset} \val_i & \Big( 1- \textstyle\prod_{j\in \ssetarg{i}} \Big( & & 1-\textstyle\sum_{k\in\kset}x_{jk}\Big) \Big) \label{eq:objF}\\ 
\text{subject to}  
& & \textstyle\sum_{j\in \sset} \cst_j x_{jk} 	& \leq B_k 	& & \forall k \in \kset\label{eq:BudgetF} \\
& & \textstyle\sum_{k\in \kset} x_{jk}		& \leq 1 		& & \forall j \in \sset\label{eq:max1AssignedF} \\
& & x_{jk}							& \in [0, 1]  	& & \forall j \in \sset, \  \forall k \in \kset \label{eq:01xF}
\end{alignat}
\end{subequations}

Obviously, a feasible solution $x$ for the problem (\lp) is also a feasible solution for (CP). 
In addition, the objective function values of (\lp) and (CP) are the same for every integral solution, i.e., $L(x) = F(x)$ for every integer solution $x$.\footnote{Our formulation (CP) is even slightly stronger than standard pipage formulations in the sense that $L(x) = F(x)$ if for all $j \in \sset$ it holds that $\sum_{k\in \kset}x_{jk} \in \mset{0,1}$.}
Moreover, for fractional solutions $x$ the value $F(x)$ is lower bounded by $L(x)$ as we show in the next lemma:

\begin{restatable}{lemma}{PipageLem}
\label{lem:boundFonL}
For every feasible solution $x$ of (\lp), we have that $F(x)\geq \left(1-\frac{1}{e}\right)L(x)$.
\end{restatable}
\begin{proof}
Let $x$ be a feasible solution for (\lp). 
Fix an item $i \in \iset$ and let $s = |\sset(i)|$ be the number of sets containing $i$. We obtain 
\begin{align*}
& 1- \prod_{j\in\ssetarg{i}}\bigg(1-\sum_{k\in\kset}x_{jk}\bigg) 
 \geq 1 - \bigg( \frac{1}{s} \sum_{j \in \ssetarg{i}} \bigg(1-\sum_{k\in\kset}x_{jk}\bigg) \bigg)^{s}\\
& \quad \ge 1-  \bigg( 1- \frac{1}{s} \sum_{j\in \ssetarg{i}} \sum_{k \in \kset}x_{jk}\bigg)^{s} \geq \left(1 - \left(1 - \frac{1}{s}\right)^{s}\right)\min \bigg\{1, \sum_{j \in \ssetarg{i}}\sum_{k \in \kset}x_{jk}\bigg\} \\
& \quad \ge \left(1 - \frac{1}{e}\right)y_i. 
\end{align*}
The first inequality follows from the arithmetic/geometric mean inequality (see, e.g., \cite{Goemans1994}) which states that for any $n$ non-negative numbers $a_1,a_2,\dots,a_n$, we have 
$\textstyle\prod_{i=1}^n a_i \le (\frac{1}{n} \sum_{i = 1}^n a_i)^n$. 
The final inequality follows from the fact that  for every $n\geq 1$ it holds that $\left(1-\frac{1}{n}\right)^n \le \frac{1}{e}$. 

Using the above, we can conclude that 
\[
F(x) =
\sum_{i \in \iset} \val_i \bigg( 1- \prod_{j\in \ssetarg{i}} \bigg( 1 - \sum_{k\in\kset} x_{jk}\bigg) \bigg) \geq  \sum_{i \in \iset} \val_i \left( 1 - \frac{1}{e} \right) y_i 
= \bigg(1 - \frac{1}{e}\bigg)L(x).
\]
\vneg
\end{proof}

In Theorem \ref{th:boundedsplit} below, we show that we can transform an optimal LP-solution $x^*_{\lp}$ into a solution $x$ that satisfies the bounded split property without decreasing the objective value of~(CP).

\begin{theorem}
\label{th:boundedsplit}
There exists a feasible solution $x$ of (\lp) that satisfies the bounded split property and for which $F(x) \geq F(x^*_{\lp})$.
\end{theorem}

Let $x$ be a feasible solution to (\lp) and let $H_x = (\sset \cup \kset, E_x)$ be the support graph of $x$.
Consider a (maximal) path $P = \seq{u_1, \dots, u_t}$ in $H_x$ that starts and ends with a node of degree one. 
We call $P$ an \emph{$\sset$-$\sset$-path} if $u_1, u_t \in \sset$, an \emph{$\sset$-$\kset$-path} if $u_1 \in \sset$ and $u_t \in \kset$, and a \emph{$\kset$-$\kset$-path} if $u_1, u_t \in \kset$.

An outline of the proof of Theorem~\ref{th:boundedsplit} is as follows: 
First, we show that there exists an optimal solution $x^* = x^*_{\lp}$ for (\lp) such that the support graph $H_{x^*}$ is acyclic (Lemma~\ref{lem:noCycles}). 
Second, we prove that from $x^*$ we can derive a solution $x'$ whose support graph $H_{x'}$ does not contain any $\sset$-$\sset$-paths such that the objective function value of (CP) does not decrease (Lemma \ref{lem:propUUpaths}). 
Finally, we show that this solution $x'$ satisfies the bounded split property (Lemma~\ref{lem:xSatisfiesBoundedSplit}).
Combining these lemmas proves Theorem~\ref{th:boundedsplit}.

\begin{restatable}{lemma}{AcyclicLem}
\label{lem:noCycles}
There is an optimal solution $x^* = x^*_{\lp}$ of (\lp) whose support graph $H_{x^*}$ is acyclic. 
\end{restatable}

\begin{proof}
Let $x$ be an optimal solution for (\lp) and suppose the support graph $H_{x} = (\sset \cup \kset, E_x)$ of $x$ contains a cycle $C = \seq{u_1, \dots, u_{t}, u_1}$ (visiting nodes $u_1, \dots, u_t, u_1$). 
Note that the length (number of edges) of $C$ is even because $H_{x}$ is bipartite. We can therefore decompose $C$ into two matchings $M_1$ and $M_2$ with $|M_1| = |M_2|$. 
Define 
\begin{align*}
& \varepsilon := 
\min \Big\{ \min_{\mset{j, k} \in M_1} \mset{\cst_{j} x_{jk}} , \min_{\mset{j, k}\in M_2}\{\cst_j(1-x_{jk})\}\Big\}.
\end{align*}
We call each edge on $C$ for which the minimum is attained a \emph{critical edge}; note that there is at least one critical edge.
We use $\varepsilon$ to define a new solution $x(\varepsilon)$ as follows:
\begin{equation}
\label{eq:xepsilon}
x_{jk}(\varepsilon) =
\begin{cases}
x_{jk} + \frac{\varepsilon}{\cst_j} & \text{if } \mset{j, k} \in M_1\\
x_{jk} - \frac{\varepsilon}{\cst_j} & \text{if } \mset{j,k} \in M_2\\
x_{jk} & \text{otherwise.}
\end{cases}
\end{equation}
In the way $\varepsilon$ is defined, the value of each critical edge $\hat{e}$ on $C$ with respect to $x(\varepsilon)$ is integral; more specifically, $x_{\hat{e}}(\varepsilon) = 0$ if $\hat{e} \in M_1$ and $x_{\hat{e}}(\varepsilon)=1$ if $\hat{e} \in M_2$. Thus, $\hat{e}$ is not part of the support graph $H_{x(\varepsilon)}$. 
As a consequence, $H_{x(\varepsilon)}$ is a subgraph of $H_x$ that has at least one cycle less. 

It remains to show that $x(\varepsilon)$ is a feasible solution to (\lp) and has the same objective function value as $x$. 
Every node $u = u_i$, $i \in [t]$, on the cycle $C$ has two incident edges in $C$, one belonging to $M_1$ and one to $M_2$. We distinguish two cases: 

Case 1: $u = k \in \kset$. 
Let the two incident edges be $\mset{i, k} \in M_1$ and $\mset{j, k} \in M_2$. 
The combined cost of these two edges in $x(\varepsilon)$ is: 
\[
c_i x_{ik}(\varepsilon) + c_j x_{jk}(\varepsilon) 
= \cst_i \left( x_{ik} + \frac{\varepsilon }{\cst_i}\right) + \cst_j \left(x_{jk} - \frac{\varepsilon }{\cst_j}\right) 
= \cst_ix_{ik} + \cst_j x_{jk}.
\]
That is, the total cost assigned to knapsack $k$ in $x(\varepsilon)$ is the same as in $x$. We conclude that $x(\varepsilon)$ satisfies constraint \eqref{eq:BudgetILP} (because $x$ does). 

Case 2: $u = j \in \sset$. 
Let the two incident edges be $\mset{j, k} \in M_1$ and $\mset{j, l} \in M_2$. By the definition of $x(\varepsilon)$, we obtain 
\[
x_{jk}(\varepsilon) + x_{jl}(\varepsilon) 
= x_{jk} + \frac{\varepsilon }{\cst_j} + x_{jl} - \frac{\varepsilon}{\cst_j} 
= x_{jk} + x_{jl}. 
\]
In particular, this implies that $x(\varepsilon)$ satisfies constraints \eqref{eq:max1assigned} and \eqref{eq:XYcoupleILP} (because $x$ does). Further, from the equation above it also follows that the solutions $x(\varepsilon)$ and $x$ result in the same objective function value of (\lp), i.e., $L(x(\varepsilon)) = L(x)$. 

By repeating the above procedure, we eventually obtain a feasible solution $x^*$ of (\lp) such that $H_{x^*}$ is a subgraph of $H_x$ that does not contain any cycles and $L(x^*) = L(x)$, i.e., $x^*$ is an optimal solution.
\end{proof}

\begin{restatable}{lemma}{NoSSPathLem}
\label{lem:propUUpaths}
There exists a feasible solution $x'$ of (\lp) whose support graph $H_{x'}$ is acyclic and does not contain any $\sset$-$\sset$-paths, and which satisfies $F(x')\geq F(x^*_{\lp})$.
\end{restatable}

\begin{proof}
Let $x$ be an optimal solution for (\lp) whose support graph $H_x$ is acyclic; we know that such a solution exists by Lemma~\ref{lem:noCycles}. 
Suppose $H_x$ contains an $\sset$-$\sset$-path $P = \seq{u_1, \dots, u_{t}}$; recall that nodes $u_1$ and $u_t$ have degree one in $H_x$. $P$ has even length because $H_x$ is bipartite.
We can thus decompose $P$ into two matchings $M_1$ and $M_2$ with $|M_1| = |M_2|$. 
We define $x(\varepsilon)$ as in \eqref{eq:xepsilon} for every $\varepsilon\in [-\varepsilon_1,\varepsilon_2]$, where 
\begin{align*}
\varepsilon_1 & := 
\min \Big\{ \min_{\mset{j, k} \in M_1} \mset{\cst_{j} x_{jk}} , \min_{\mset{j, k}\in M_2}\{\cst_j(1-x_{jk})\}\Big\} \\
\varepsilon_2 & := 
\min \Big\{ \min_{\mset{j, k} \in M_1} \mset{\cst_{j} (1-x_{jk})} , \min_{\mset{j, k}\in M_2}\{\cst_j x_{jk}\}\Big\}.
\end{align*}

We first show that $x(\varepsilon)$ is feasible solution for (CP) for every $\varepsilon\in [-\varepsilon_1,\varepsilon_2]$. By following the same line of arguments as in the proof of Lemma~\ref{lem:noCycles}, we can show that the solution $x(\varepsilon)$ satisfies constraint (\ref{eq:BudgetF}). Further, for every $j \in \sset$, $j \neq u_1, u_t$, we can also show (as in the proof of Lemma~\ref{lem:noCycles}) that $\sum_{k \in \kset}x_{jk}(\varepsilon) = \sum_{k \in \kset}x_{jk}$ and thus constraints (\ref{eq:max1AssignedF}) and (\ref{eq:01xF}) are satisfied in $x(\varepsilon)$.
By definition, the endpoints $j = u_1, u_t$ of $P$ have degree one in $H_x$ and thus $j$ is fractionally assigned to a single knapsack in $x$. 
Consider $j = u_1$ and let $k = u_2$ be the knapsack to which $j$ is assigned in $x$. We have $x_{jk} \in (0,1)$ and, from the definition of $\varepsilon_1$ and $\varepsilon_2$, it follows that $x_{jk}(\varepsilon) \in [0, 1]$. The same argument holds for $j = u_t$. Thus, $x(\varepsilon)$ also satisfies constraints (\ref{eq:max1AssignedF}) and (\ref{eq:01xF}) for $j = u_1, u_t$ if $\varepsilon \in [-\varepsilon_1,\varepsilon_2]$. 
We conclude that $x(\varepsilon)$ is a feasible solution for (CP) if $\varepsilon \in [-\varepsilon_1,\varepsilon_2]$.

Next, we show that $F(x(-\varepsilon_1))$ or $F(x(\varepsilon_2))$ is at least as large as $F(x)$. To this aim, we show that $F(x(\varepsilon))$ as a function of $\varepsilon$ is convex; in fact, we show convexity for each item $i \in \iset$ separately. 
Observe that the first and the last edge of $P$ are in different matchings. Without loss of generality, we assume that $\mset{u_1, u_2} \in M_1$ and $\mset{u_{t-1}, u_t} \in M_2$. 
The contribution of $i$ to the objective function $F(x(\varepsilon))$ can be written as $\val_i f_i(\varepsilon)$, where
\begin{align*}
f_i(\varepsilon) 
& = \bigg(1-\prod_{j\in \ssetarg{i}}\Big(1-\sum_{k\in \kset}x_{jk}(\varepsilon)\Big)\bigg) \\
& =  \bigg(1- \prod_{j \in \ssetarg{i} \cap \mset{u_1, u_t}} \Big(1-\sum_{k\in \kset}x_{j k}(\varepsilon)\Big) 
\prod_{j\in \ssetarg{i} \setminus \mset{u_1,u_t}}\Big(1-\sum_{k\in \kset}x_{jk}\Big)\bigg).
\end{align*}
(Here, we adopt the convention that the empty product is defined to be $1$.)
Note that the latter product $\prod_{j\in \ssetarg{i} \setminus \mset{u_1, u_t}} (1-\sum_{k\in \kset}x_{jk})$ is independent of $\varepsilon$ and has a value between 0 and 1. 
We now distinguish three cases:

Case 1: $|\mset{u_1, u_t} \cap \ssetarg{i}| = 2$.
In this case, we have 
\[
\prod_{j \in \ssetarg{i} \cap \mset{u_1, u_t}} \Big(1-\sum_{k\in \kset}x_{j k}(\varepsilon)\Big) 
= 
\Big( 1-\big(x_{u_1u_2}+\frac{\varepsilon}{\cst_{u_1}}\big) \Big)
\Big( 1-\big(x_{u_t u_{t-1}}-\frac{\varepsilon}{\cst_{u_{t}}}\big) \Big) 
\]
That is, $f_i(\varepsilon)$ is a quadratic function of $\varepsilon$ and the coefficient of the quadratic term is $1/(\cst_{u_1}\cst_{u_t})$, which is positive. Thus, $f_i(\varepsilon)$ is convex. 

Case 2: $|\mset{u_1, u_t} \cap \ssetarg{i}| = 1$. 
In this case, $f_i(\varepsilon)$ is a linear function in $\varepsilon$. 

Case 3: $|\mset{u_1, u_t} \cap \ssetarg{i}| = 0$. 
In this case, $f_i(\varepsilon)$ is independent of $\varepsilon$.

We conclude that $F(x(\varepsilon))$ is convex in $\varepsilon$ and its maximum over $[-\varepsilon_1,\varepsilon_2]$ is thus attained at one of the endpoints, i.e., 
$
\max\{F(x(-\varepsilon_1)), F(x(\varepsilon_2))\} 
= \max_{\varepsilon\in [-\varepsilon_1,\varepsilon_2]}\{F(x(\varepsilon))\} \geq F(x(0)) 
= F(x).
$

As a result, we can find a feasible solution $x' \in \mset{x(-\varepsilon_1), x(\varepsilon_2)}$ with the property that $F(x')\geq F(x)$.
Further,  $x'$ has at least one fractional variable on $P$ less than $x$. Thus, $H_{x'}$ is a subgraph of $H_x$ with at least one edge of $P$ removed. By repeating this procedure, we eventually obtain a feasible solution $x'$ of (CP) whose support graph does not contain any $\sset$-$\sset$-paths, and for which $F(x')\geq F(x)$. 
\end{proof}

\begin{restatable}
{lemma}{RoundingLem}
\label{lem:xSatisfiesBoundedSplit}
Let $x'$ be a feasible solution of (\lp) whose support graph $H_{x'}$ is acyclic and does not contain any $\sset$-$\sset$-paths. Then $x'$ satisfies the bounded split property. 
\end{restatable}

\begin{proof}
We have to prove that there exists a matching $M$ in $H_{x'}$ that saturates all (non-isolated) nodes in $\sset$. For simplicity, we can assume in this proof that all isolated nodes are removed from $H_{x'}$. 
Note that we assume that $H_{x'}$ is acyclic and thus a forest. The idea is to construct a matching $M_T$ for each tree $T$ of the forest $H_{x'}$. Our final matching $M$ is then simply the union of all these matchings, i.e., $M = \cup_T M_T$. 

Consider a tree $T$ of $H_{x'}$ and root it at an arbitrary node $r \in \kset$. By assumption, $T$ has at most one leaf in $\sset$ because otherwise there would exist an $\sset$-$\sset$-path in $H_{x'}$. 
We can construct a matching $M_T$ that matches all the nodes in $T \cap \sset$ as follows:\footnote{We slightly abuse notation here by letting $T$ refer to the tree and the set of nodes it spans.} 
If there is a (unique) $\sset$-leaf, say $j \in \sset$, then we match each $\sset$-node on the path from $j$ to the root $r$ to its unique parent in $T$. Each remaining $\sset$-node in $T$ is matched to one of its children (chosen arbitrarily); there always is at least one child as $j$ is the only $\sset$-leaf in $T$. Note that this defines a matching $M_T$ that matches all nodes in $T \cap \sset$ as desired. 

\end{proof}

\begin{algorithm2e}[t]
\caption{$\frac{1}{2}(1-\frac{1}{e})$-approximation algorithm for MCPK.}
\label{alg:apxAlg}
Compute an optimal solution $x^*$ to (\lp). \;
Derive a solution $x$ from $x^*$ satisfying the bounded split property (Theorem~\ref{th:boundedsplit}). \;
Let $M$ be the corresponding $\sset$-saturating matching. \;
Decompose the fractional solution $x$ into $x^1, x^2$ as follows:
\vspace*{-1.5ex}
\[
x^1_{jk} = 
\begin{cases}
x^1_{jk} = 1 & \text{if $x_{jk} = 1$} \\
x^1_{jk} = 0 & \text{otherwise} 
\end{cases} 
\quad\;\;\text{and}\quad\;\;
x^2_{jk}  = 
\begin{cases}
x^2_{jk} = 1 & \text{if $x_{jk} \in (0,1)$, $\mset{j,k} \in M$} \\
x^2_{jk} = 0 & \text{otherwise}.
\end{cases} 
\] 
\vspace*{-4.5ex}
\;
Output $x_{\alg} \in \arg\max \mset{L(x^1), L(x^2)}$ \;

\end{algorithm2e}

In light of Theorem~\ref{th:boundedsplit}, our approximation algorithm (Algorithm \ref{alg:apxAlg}) simply chooses the better of the integral and the rounded solution obtained from an optimal LP-solution.

\begin{restatable}{theorem}{ApxThmMCPK}
\label{th:apxalg}
Algorithm \ref{alg:apxAlg} is a $\frac{1}{2}\left(1-\frac{1}{e}\right)$-approximation algorithm for MCPK. 
\end{restatable}
\begin{proof}
Note that all the procedures described in the proof of Theorem~\ref{th:boundedsplit} to derive a solution $x$ from $x^*$ can be implemented to take polynomial time. The running time of Algorithm \ref{alg:apxAlg} is thus polynomial. 

The solution $x^1$ constructed by the algorithm corresponds to the integral part of the solution $x$ satisfying the bounded split property. Clearly, this is a feasible solution to MCPK. Further, the solution $x^2$ is derived from the fractional part of $x$ using the $\sset$-saturating matching $M$. Note that by construction (see also Footnote~\ref{rem:implicit-restrictions}) each set $j \in \sset$ is matched to some $k \in \kset$ with $\bgt_k \ge \cst_j$. Thus, $x^2$ is a feasible integral solution. 

It remains to bound the approximation factor of the algorithm. We have 
\begin{align*}
L(x_{\alg}) 
& = \max\{L(x^1),L(x^2)\}
  \geq \frac{1}{2}(L(x^1)+L(x^2)) 
  = \frac{1}{2}(F(x^1)+F(x^2)) 
 \geq \frac{1}{2}F(x) \\
 & \geq \frac{1}{2}F(x^*) 
 \geq \frac{1}{2}\left(1-\frac{1}{e}\right)L(x^*)
 \geq \frac{1}{2}\left(1-\frac{1}{e}\right) \opt.
\end{align*}
Here, the second equality follows from the fact that the objective function values of (\lp) and (CP) are the same for integer solutions (as observed above). 
To see that the second inequality holds, note that by the definitions of $x^1$ and $x^2$ we have 
\[
\sum_{j \in\sset}\sum_{k \in \kset}x_{jk} \leq \sum_{j \in\sset}\sum_{k \in \kset} (x^1_{jk}+x^2_{jk})
\]
and the function $F(x)$ is non-decreasing in $x$. 
The third inequality follows from Lemma~\ref{lem:propUUpaths} and 
the fourth inequality holds because of Lemma~\ref{lem:boundFonL}. 
The final inequality holds because (\lp) is a relaxation of the integer program (\ip) of MCPK.
\end{proof}

\section{Maximum Coverage Problem with Cluster Constraints}
\label{ss:MCPC}

We derive an approximation algorithm for the Maximum Coverage Problem with Cluster Constraints (\probl). Our algorithm exploits the existence of an LP-based approximation algorithm for the problem without cluster constraints. In particular, we use our algorithm for MCPK derived in the previous section as a subroutine to obtain a $\frac13 (1- \frac{1}{e})$-approximation algorithm for \probl.

A key element of our approach is to integrate the cluster capacities into the IP formulation of MCPK by introducing a variable $z_{kl}$ for every cluster $l \in \cset$ and every knapsack $k \in \ksetarg{l}$, which specifies the fraction of the cluster capacity $\ccap_l$ that is assigned to knapsack $k$. 
We obtain the following MIP formulation for \probl: 
\begin{subequations}
\begin{alignat}{3}
\max \quad 
& & L(x,z) = \sum_{i \in \iset} \val_iy_i 						&				& & 									\label{eq:objILPMCPC} \\
\text{subject to} \quad 
& & \eqref{eq:BudgetILP}\! & - \!\eqref{eq:intILP} & &  \\
& & \textstyle \sum_{j\in \sset} \cst_jx_{jk} 					& \leq \ccap_lz_{kl} \quad 	& & \forall l \in \cset,\ \forall k \in \ksetarg{l}		\label{eq:BudgetKnapsackFromClusterILPMCPC} \\
& & \textstyle \sum_{k\in \ksetarg{l}}z_{kl}					& \leq 1 			& & \forall l \in \cset 						\label{eq:FractionKnapsackFromClusterILPMCPC} \\
& & z_{kl} 											& \geq 0 			& & \forall l \in \cset,\ \forall k \in \ksetarg{l}	 \label{eq:posZMCPC}
\end{alignat}
\end{subequations}

Note that constraints \eqref{eq:BudgetKnapsackFromClusterILPMCPC}--\eqref{eq:FractionKnapsackFromClusterILPMCPC} ensure that the capacity of cluster $l \in \cset$ is not exceeded. 
The advantage of this formulation is that, once we know the optimal values of the $z_{kl}$-variables, the remaining problem basically reduces to an instance of MCPK (though a subtle point remains, as explained below). 
We use (\ip) and (\lp) to refer to the integer formulation above and its relaxation, respectively. Let $\opt$ refer to the objective function value of an optimal solution to \probl.
Note that also for this formulation, the remarks that are given in Footnote~\ref{rem:implicit-restrictions} apply.
Similar to MCPK, for every feasible fixing of the $x_{jk}, z_{kl}$-variables, 
the optimal $y_i$-variables can be determined as before. 
As a consequence, an optimal solution $(x, y, z)$ of the above program is fully specified by $(x, z)$. 

It will be convenient to assume that the knapsacks $\ksetarg{l} = \mset{1, \dots, \nkincarg{l}}$ of each cluster $l \in \cset$ are ordered by non-increasing capacities (breaking ties arbitrarily), i.e., if $k, k' \in \ksetarg{l}$ with $k < k'$ then $\bgt_k \ge \bgt_{k'}$. 
The following notion will be crucial: 
The knapsack $\crit{l} \in \ksetarg{l}$ which satisfies $\sum_{k=1}^{\crit{l}-1}\bgt_{k}\leq \ccap_l <  \sum_{k=1}^{\crit{l}}\bgt_{k}$ is called the \emph{critical knapsack} of cluster $l$.\footnote{Note that there always exists such a critical knapsack by the assumption that $\sum_{k\in \ksetarg{l}}\bgt_k>\ccap_l$.}

We first show that the $z$-variables of an optimal LP-solution admit a specific structure. Intuitively, the lemma states that the cluster capacity $\ccap_l$ of each cluster $l$ is shared maximally among the first $\crit{l}-1$ knapsacks in $\ksetarg{l}$ and the remaining capacity is assigned to the critical knapsack $\crit{l}$.

\begin{restatable}{lemma}{OptZLem}
\label{lem:optZ}
There is an optimal solution $(x^*, z^*)$ of (\lp) such that for every cluster $l \in \cset$,
$z^*_{kl} = {\bgt_k}/{\ccap_l}\ $ for $k < \crit{l}$, 
$z^*_{kl} = 1- \sum_{t = 1}^{\crit{l} - 1}z^*_{tl}\ $ for $k = \crit{l}$ and 
$z^*_{kl} = 0$ otherwise. 
\end{restatable}
\begin{proof}
First of all, we show that $z^*$ as defined above is feasible (i.e., satisfies \eqref{eq:BudgetKnapsackFromClusterILPMCPC}, \eqref{eq:FractionKnapsackFromClusterILPMCPC} and \eqref{eq:posZMCPC}). 
Fix some cluster $l \in \cset$. Clearly, for all $k \neq \crit{l}$ we have $z^*_{kl} \ge 0$ and by definition $\sum_{k=1}^{\nkincarg{l}} z_{kl}^*=1$. It remains to show that $z_{\crit{l}l}\geq 0$ or, equivalently, $\sum_{k = 1}^{\crit{l}-1} z_{kl}^*\leq 1$. The latter follows by exploiting the definition above and the fact that $\crit{l}$ is the critical knapsack of $l$ and thus $\sum_{k = 1}^{\crit{l} - 1} z_{kl}^*  = \sum_{k = 1}^{\crit{l} - 1} {\bgt_k}/{\ccap_l} \leq 1$. 

Next, we show that any optimal solution $(x, z)$ can be transformed into an optimal solution $(x^*, z^*)$, where $z^*$ is defined as above. We argue cluster by cluster. Fix some cluster $l \in \cset$ and let $\ksetarg{l} = \mset{1, \dots, \nkincarg{l}}$ be the (ordered) set of knapsacks.
{First of all, we assume without loss of generality that there is no knapsack $k$ for which $z_{kl}>{\bgt_k}/{\ccap_l}$ because the cost assigned to knapsack $k$ cannot exceed $\bgt_k$. Furthermore, we assume that $\sum_{k\in \ksetarg{l}} z_{kl}=1$ because we know that a solution $(x,z)$ for which $\sum_{k\in \ksetarg{l}} z_{kl}<1$ cannot be an optimal solution.}
Let $k \in \ksetarg{l}$ be the first knapsack (in this order) satisfying $z_{kl} < z^*_{kl}$. (If no such knapsack exists, we are done because $z_{kl} \ge z^*_{kl}$ implies $z_{kl} = z^*_{kl}$ for all $k$ by the definition of $z^*$.) Let $k' > k$ be the last knapsack satisfying $k' > k$ with $z_{k'l} > z^*_{k'l}$. 
{We know that such knapsack exists because $\sum_{k \in \ksetarg{l}}z_{kl} = 1$ and any knapsack $k'<k$ has the property that $z^*_{kl} = {\bgt_k}/{\ccap_l}$. Hence, there has to be a knapsack $k'\in \ksetarg{l}$ with $z_{k'l} > z^*_{k'l}$ and that knapsack $k'$ has to have a larger index than $k$.}

Let $\Delta(k,k') = \min\{(z^*_{kl} - z_{kl})\bgt_k, (z_{k'l} - z^*_{k'l}) \bgt_{k'}\}$. Given that the cost $\cst_j$ of each set $j \in \sset$ is independent of the knapsack to which it is assigned to and $\bgt_{k'} \ge \bgt_k$ by our ordering, we can reassign a total contribution (in terms of cost) of $\Delta(k,k')$ units from knapsack $k'$ to $k$ by changing some $x_{jk},x_{jk'}$-variables accordingly. Note that this shift is feasible because every set $j \in \sset$ that is fractionally assigned to knapsack $k'$ also fits on knapsack $k$ (by our ordering). Further, this shift does not change any of the $y$-variables and thus the objective function value remains the same. By continuing this way, we eventually obtain a feasible solution $(x^*, z^*)$ which is also optimal. 
\end{proof}
An approach that comes to one's mind is as follows: Fix the $z^*$-values as in Lemma~\ref{lem:optZ} and let $\lp(z^*)$ be the respective LP-relaxation.
Note that $\lp(z^*)$ is basically the same as the LP-relaxation of MCPK, where each knapsack $k \in \ksetarg{l}$, $l \in \cset$, has a \emph{reduced capacity} of $\min\mset{\bgt_k, \ccap_l z^*_{kl}}$. 
So we could compute an optimal solution $x^*$ to $\lp(z^*)$ and use our LP-based approximation algorithm (Algorithm~\ref{alg:apxAlg}) to derive an integral solution $x_{\alg}$ 
satisfying $L(x_{\alg}) \ge \frac12(1-\frac{1}{e}) L(x^*, z^*) \ge \frac12(1-\frac{1}{e}) \opt$. 
Unfortunately, however, this approach fails because of the following subtle point: 
If a set $j \in \sset$ is fractionally assigned to some critical knapsack $k \in \kset$ in the optimal LP-solution $x^*$ of $\lp(z^*)$, then it might be infeasible to assign $j$ to $k$ integrally. We could exclude these infeasible assignments beforehand (i.e., by setting $x_{jk} = 0$ whenever $\cst_j > \ccap_l z^*_{kl}$ for a critical knapsack $k$), but then an optimal LP-solution might not recover a sufficiently large fraction of $\opt$.

Instead, we can fix this problem by using a slightly more refined algorithm described in Algorithm~\ref{alg:apxAlgMCPC}. We decompose the fractionally assigned sets into two solutions, one using non-critical knapsacks and one using critical knapsacks only, and then choose the better of those and the integral solution. 
\begin{algorithm2e}[t]
\caption{$\frac{1}{3}(1-\frac{1}{e})$-approximation algorithm for the \probl.}
\label{alg:apxAlgMCPC}
Fix $z^*$ as in Lemma \ref{lem:optZ} and compute an optimal solution $x^*$ to $(\lp(z^*))$. \;
Derive a solution $x$ from $x^*$ that satisfies the bounded split property (Theorem~\ref{th:boundedsplit}). \;
Let $M$ be the corresponding $\sset$-saturating matching. \;
Decompose the fractional solution $x$ into $x^1, x^2, x^3$ as follows: 
\vspace*{-1.5ex}
\begin{align*}
x^1_{jk} & = 
\begin{cases}
x^1_{jk} = 1 & \text{if $x_{jk} = 1$} \\
x^1_{jk} = 0 & \text{otherwise} \\
\end{cases} \\
x^2_{jk} & = 
\begin{cases}
x^2_{jk} = 1 & \text{if $x_{jk} \in (0,1)$, $k$ not critical, $\mset{j,k} \in M$} \\
x^2_{jk} = 0 & \text{otherwise} \\
\end{cases} \\
x^3_{jk} & = 
\begin{cases}
x^3_{jk} = 1 & \text{if $x_{jk} \in (0,1)$, $k$ critical, $\mset{j,k} \in M$} \\
x^3_{jk} = 0 & \text{otherwise} \\
\end{cases}
\end{align*}
\vspace*{-4.5ex}
\;
Output $x_{\alg} \in \arg\max \mset{L(x^1), L(x^2), L(x^3)}$. \;
\end{algorithm2e}

\begin{restatable}{theorem}{ApxMCPCThm}
\label{th:apxAlgMCPC}
Algorithm \ref{alg:apxAlgMCPC} is a $\frac{1}{3}(1-\frac{1}{e})$-approximation algorithm for \probl.
\end{restatable}
\begin{proof}
Note that by Theorem~\ref{th:boundedsplit} the fractional solution $x$ derived in Step 2 of the algorithm is a feasible solution to ($\lp(z^*)$). Clearly, $x^1$ is a feasible (integral) solution. Further, the $\sset$-saturating matching $M$ ensures that $\cst_j \le \bgt_k$ for every $\mset{j, k} \in M$. In particular, this implies that the solution $x^2$ is feasible because for every (non-critical) knapsack $k \in \ksetarg{l}$, $l \in \cset$, we have $\bgt_k = \ccap_l z^*_{kl}$ by Lemma~\ref{lem:optZ}, and thus 
$\sum_{j \in \sset} c_j x^2_{jk} \le B_k = \ccap_l z^*_{kl}$. 

It remains to argue that $x^3$ is feasible. But this holds because for every cluster $l \in \cset$ there is at most one set $j \in \sset$ assigned to this cluster, namely the one (if any) assigned to the critical knapsack $\crit{l} \in \ksetarg{l}$ with $\mset{j, k} \in M$. In particular, for $k = \crit{l}$ we have $\sum_{j \in \sset} \cst_j x^3_{jk} \le \bgt_k \le \ccap_l$.

It remains to bound the approximation factor of the algorithm. Using the same arguments as in the proof of Theorem~\ref{th:apxalg}, we obtain 
\begin{align*}
L(x_{\alg}) 
& \geq \frac{1}{3}(L(x^1)+L(x^2)+L(x^3)) 
  = \frac{1}{3}(F(x^1)+F(x^2)+F(x^3))
   \geq \frac{1}{3}F(x)\\
& \geq \frac{1}{3}F(x^*)
\geq \frac{1}{3}\bigg(1-\frac{1}{e}\bigg)L(x^*)
= \frac{1}{3}\bigg(1-\frac{1}{e}\bigg)L(x^*, z^*) 
\geq \frac{1}{3}\bigg(1-\frac{1}{e}\bigg) \opt.
\end{align*}
Note that the last equality holds because $x^*$ is an optimal solution to ($\lp(z^*)$). 
\end{proof}

\section{Multiple Knapsack Problem with Cluster Constraints}
\label{s:additionalApplications}

Our technique introduced in the previous section 
can also be applied to other cluster capacitated problems. One such example is the Multiple Knapsack Problem with Cluster Constraints (MKPC). 
We first derive a $\frac{1}{3}$-approximation algorithm for MKPC and then present a more sophisticated iterative rounding scheme that provides a $\frac{1}{2}$-approximation algorithm for certain special cases of MKPC. 

Note that for MKPC the notions of sets and items coincide and we simply refer to them as items; in particular, each item $j \in \sset$ now has a profit $\val_j$ and a cost $\cst_j$. Thus, we can also drop the $y$-variables in the ILP formulation of the problem.
Throughout this section, we assume that the knapsacks in $\kset = [p]$ are ordered by non-increasing capacities, i.e., if $k, k' \in \kset$ with $k < k'$ then $\bgt_k \ge \bgt_{k'}$.

\subsection{$\frac{1}{3}$-Approximation for MKPC with General Clusters}\label{ss:MKP}

MKPC is a generalization of the classical multiple knapsack problem (MKP). 
For MKP, the optimal solution of the LP-relaxation satisfies the bounded split property (see, e.g., \cite{Shmoys1993}) and there is a natural greedy algorithm to find an optimal solution of the LP-relaxation (see, e.g., \cite{Kellerer2004}).
Here we exploit some ideas of the previous section to derive a greedy algorithm (called \greedy\ subsequently) for MKPC and prove that it computes an optimal solution to the LP-relaxation; this algorithm is also used at the core of our iterative rounding scheme in the next section. 
Further, we show that the constructed solution satisfies the bounded split property. Exploiting this, we can then easily obtain a $\frac{1}{3}$-approximation algorithm for MKPC.

Our algorithm first fixes optimal $z^*$-variables as defined in Lemma~\ref{lem:optZ} and then runs an adapted version of the greedy algorithm in \cite{Kellerer2004} on the instance with the reduced capacities.\footnote{The greedy algorithm for MKP described in \cite{Kellerer2004} operates on a per-knapsack basis, while our algorithm proceeds on a per-item basis.} 
 \greedy\ assigns items in non-increasing order of their \emph{efficiency ratios} (breaking ties arbitrarily), where the efficiency ratio of item $j \in \sset$ is defined as $\val_j/\cst_j$. 
When item $j$ is considered, it is assigned to the knapsack with the smallest capacity that can hold the item and has some residual capacity. More formally, a knapsack $k \in \kset$ \emph{can hold} item $j \in \sset$ if $\bgt_k \ge \cst_j$. 
Further, we say that a knapsack $k$ has \emph{residual capacity} with respect to $(x^*, z^*)$ if 
$
\res_k(x^*, z^*) := \ccap_{\csetarg{k}}z^*_{k\csetarg{k}}-\sum_{j\in \sset}\cst_j x^*_{jk} > 0.
$
(Recall that $\csetarg{k}$ denotes the cluster to which knapsack $k$ belongs.) 
We continue this way until either item $j$ is assigned completely (possibly split over several knapsacks) or all knapsacks that can hold $j$ have zero residual capacity. We then continue with the next item in the order.

We show that \greedy\ computes an optimal fractional solution which satisfies the bounded split property.
We say that an item $j \in \sset$ is a \emph{split item} if it is fractionally assigned to one or multiple knapsacks. Let $\kcs$ refer to the set of all split items. 
An item $j \in \kcs$ is a \emph{split item of knapsack $k \in \kset$} if $k$ is the knapsack with the smallest capacity to which $j$ is assigned in \greedy.

\begin{restatable}{lemma}{GreedyLem}
\label{lem:optimalityConstructiveLPRelax}
\label{lem:constructiveLPSolBoundedSplit}
\greedy\ computes an optimal solution $(x^*,z^*)$ to the LP-relaxation of MKPC which satisfies the bounded split property.
\end{restatable}

In the analysis of Algorithm \ref{alg:apxAlgMCPC} for the \probl, a factor $(1-\frac{1}{e})$ is lost in the approximation guarantee by transforming the optimal solution of the LP-relaxation to a solution that satisfies the bounded split property. By Lemma~\ref{lem:optimalityConstructiveLPRelax}, we can avoid this loss here.

\begin{theorem}
Algorithm \ref{alg:apxAlgMCPC} is a $\frac{1}{3}$-approximation algorithm for  MKPC.
\end{theorem}

\subsection{$\frac{1}{2}$-Approximation for MKPC with Isolation Property}
\label{ss:non-overlappingclusters}

We derive an iterative rounding $\frac{1}{2}$-approximation algorithm for instances of MKPC that satisfy a certain \emph{isolation property}. A practical situation in which the isolation property hold is when standardized containers that all have the same capacity are used.

Let $(x^*, z^*)$ be the optimal solution to the LP-relaxation of MKPC computed by \greedy.
We say that an item $j \in \sset$ is an \emph{unsplit item} if it is integrally assigned to some knapsack, i.e., $x^*_{jk} = 1$ for some $k \in \kset$. 
Let $\kcu$ and $\kcs$ refer to the sets of unsplit and split items, respectively. 
As argued in the proof of Lemma~\ref{lem:constructiveLPSolBoundedSplit}, the split item $j \in \kcs$ of knapsack $k \in \kset$ is unique; we use $\splits_k = j$ to refer to the split item of knapsack $k$. 
Further, we use $\usarg{k}$ to denote the set of all unsplit items assigned to knapsack $k$. For each cluster $l \in \cset$, we denote by $\usarg{l} = \cup_{k \in \ksetarg{l}} \usarg{k}$ the set of all unsplit items and by $\ssarg{l} = \cup_{k \in \ksetarg{l}} \splits_k$ the set of all split items assigned to $l$.

A cluster $\iso \in \cset$ is said to be \emph{isolated} if for every item $j \in \usarg{l}\cup \ssarg{l}$ assigned to some cluster $l \neq \iso$ it holds that $x^*_{jk} = 0$ for all $k \in \ksetarg{\iso}$.

We say that a class of instances of MKPC satisfies the \emph{isolation property} if after the removal of an arbitrary set of clusters there always exists an isolated cluster. 
For example, the isolation property holds true for instances whose clusters can be \emph{disentangled} in the sense that we can impose an order on the set of clusters $\cset = [q]$ such that for any two clusters $l, l' \in \cset$ with $l < l'$ it holds that $\min_{k\in \ksetarg{l}}\mset{\bgt_{k}}\geq\max_{k\in \ksetarg{l'}}\mset{\bgt_{k}}$.

The following Rounding Lemma provides a crucial building block of our iterative rounding scheme. 
It shows that we can always find a feasible assignment $\sigma$ which recovers at least half of the fractional profit of an isolated cluster $\iso$.
Recall that 
$\sset_{\sigma}$ refers to the set of items assigned under $\sigma$. 

\begin{restatable}[Rounding Lemma]{lemma}{IterativeRoundingLem}
\label{th:halfClusterQMKPC}
Let $\iso$ be an isolated cluster. Then there exists a feasible assignment $\sigma:\usarg{\iso}\cup\ssarg{\iso} \to \ksetarg{\iso}$ such that
\begin{equation}\label{eq:profitInequialityClusterQ}
    \sum_{j\in \sset_{\sigma}} \val_j
    \geq 
    \sum_{j\in \sset_{\sigma}} \sum_{k\in \ksetarg{\iso}}\val_j x_{jk}^*
    \ge
    \frac{1}{2} \bigg(\sum_{j \in \sset}\sum_{k\in \ksetarg{\iso}}\val_jx_{jk}^*\bigg).
\end{equation}
\end{restatable}

\begin{algorithm2e}[t]
\caption{Iterative rounding algorithm for MKPC with isolation property.}
\label{alg:alg-mkpc}
Let $\lp^{(1)}$ be the original LP-relaxation of MKPC. \;
\For{$i = 1, \dots, q$}{
Compute an optimal solution $(x^{(i)}, z^{(i)})$ to $\lp^{(i)}$ using \greedy. \;
Identify an isolated cluster $\iso^{(i)}$ with respect to $(x^{(i)}, z^{(i)})$. \;
Obtain a feasible assignment $\sigma^{(i)}$ for cluster $\iso^{(i)}$ using the Rounding Lemma. \;
Add the constraints \eqref{eq:iterativeRounding1}--\eqref{eq:itererattiveRounding0} fixing the assignment $\sigma^{(i)}$ for $\iso^{(i)}$ to obtain $\lp^{(i+1)}$.
}
\end{algorithm2e}

We next describe our iterative rounding scheme (see Algorithm~\ref{alg:alg-mkpc}). 
We first compute an optimal solution $(x^*, y^*)$ to the LP relaxation of MKPC using \greedy.  
Because of the isolation property, there exists an isolated cluster $\iso$. 
We then apply the Rounding Lemma to obtain an assignment $\sigma$ for cluster $\iso$. 
After that, we fix the corresponding variables in the LP-relaxation accordingly and repeat. By iterating this procedure, we obtain a sequence of isolated clusters $\iso^{(1)}, \dots, \iso^{(q)}$ and assignments $\sigma^{(1)}, \dots, \sigma^{(q)}$, where $\sigma^{(l)}$ is the assignment obtained by applying the Rounding Lemma to cluster $\iso^{(l)}$.

Let $\sigma(i) = \seq{\sigma^{(1)}, \dots, \sigma^{(i)}}$ be the (combined) assignment for the first $i$ clusters $\iso^{(1)}, \dots, \iso^{(i)}$ that we obtain 
at the end of iteration $i$. The LP-relaxation $\lp^{(i+1)}$ that we solve in iteration $i+1$ is then defined as the LP (\ref{eq:objILPMCPC})--(\ref{eq:posZMCPC}) with the following additional constraints:
\begin{subequations}
\begin{alignat}{2}
     x_{jk} = 1 &\quad \forall l\in\mset{\iso^{(1)}, \dots, \iso^{(i)}} \quad\forall k \in \ksetarg{l} \quad \forall j \in \sset_{\sigma(i)}(k),\label{eq:iterativeRounding1}\\
     x_{jk} = 0 & \quad\forall l\in\mset{\iso^{(1)}, \dots, \iso^{(i)}} \quad\forall k \in \ksetarg{l} \quad \forall j \not\in \sset_{\sigma(i)}(k)\label{eq:itererattiveRounding0}.
\end{alignat}
\end{subequations}
Let $(x^{(i+1)}, z^{(i+1)})$ be the optimal solution of $\lp^{(i+1)}$ computed by \greedy\ in iteration $i+1$.
The next lemma establishes that the final assignment recovers at least half of the optimal solution.

\begin{restatable}{lemma}{HalfRecoveryLem}
\label{th:iterativeRounding}
Fix some $1 \le i \le q$ and let $\sigma(i)$ be the assignment at the end of iteration $i$. Further, let $(x^*, z^*)$ be the optimal solution to the original LP-relaxation constructed by \greedy. Then

\begin{equation}
\sum_{j\in \sset_{\sigma(i)}} \val_j \geq \frac{1}{2}\sum_{l \in \mset{\iso^{(1)}, \dots, \iso^{(i)}}} \sum_{k\in \ksetarg{l}}\sum_{j\in \sset}\val_jx^*_{jk}.    
\end{equation}
\end{restatable}

We summarize the result in the following theorem. 

\begin{restatable}{theorem}{ApxThmMKPC}
Algorithm~\ref{alg:alg-mkpc} is a $\frac{1}{2}$-approximation algorithm for instances of MKPC satisfying the isolation property. 
\end{restatable}
\begin{proof}
By Lemma~\ref{th:iterativeRounding}, the final assignment $\sigma = \sigma(q)$ returned by the algorithm has profit at least 
\[
 \sum_{j\in S_\sigma}\val_j 
 \geq  \frac{1}{2}\sum_{l \in \cset}\sum_{k \in \ksetarg{l}} \sum_{j\in \sset}\val_j x^*_{jk}
 \geq \frac{1}{2} \opt. 
\]
\vneg
\end{proof}

\bibliographystyle{abbrv}
\bibliography{bibliography}
\end{document}